\documentclass[11pt,reqno]{amsart}

\usepackage{geometry}
\usepackage{hyperref}
\hypersetup{                                
    linkcolor  = blue!85!black,             
    citecolor  = red!55!black,            
    urlcolor   = red!85!black,              
    colorlinks = true,                      
} 
\usepackage{amsmath,amsthm,amssymb}
\usepackage{graphicx}
\usepackage{cleveref}
\usepackage{mathabx}
\usepackage{color}
\usepackage{todonotes}

\newtheorem{theorem}{Theorem}
\newtheorem{prop}{Proposition}
\newtheorem{lemma}{Lemma}
\newtheorem{cor}{Corollary}
\newtheorem{remark}{Remark}
\newtheorem{defn}{Definition}

\newcommand{\coup}{{g}}

\newcommand{\bx}{{\bf x}}
\newcommand{\by}{{\bf y}}

\newcommand{\eps}{\epsilon}

\newcommand{\tr}{\mathrm{Tr} \, }

\newcommand{\R}{\mathbb{R}}
\newcommand{\N}{\mathbb{N}}
\newcommand{\Z}{\mathbb{Z}}
\newcommand{\C}{\mathbb{C}}
\newcommand{\E}{\mathbb{E}}
\newcommand{\Prob}{\mathbb{P}}

\newcommand{\ssubset}{\subset\joinrel\subset}

\title{Dynamic One Photon Localization in a Discrete Model of Quantum Optics}

\author{Joseph Kraisler}
\address{Department of Mathematics and Statistics, Amherst College, Amherst, MA 01002}
\email{jkraisler@amherst.edu}

\author{Jeffrey Schenker}
\address{Department of Mathematics, Michigan State University, East Lansing, MI 48824}
\email{schenke6@msu.edu}

\author{John C. Schotland}
\address{Department of Mathematics and Department of Physics, Yale University, New Haven, CT 06511}
\email{john.schotland@yale.edu}

\date{\today}

\begin{document}

\begin{abstract}
We consider a recently proposed model for the propagation of one-photon states in a random medium of two-level atoms. We demonstrate the existence of Anderson localization of single photon states in an energy band centered at the resonant energy of the atoms. Additionally, for a Bosonic model of the atoms the results can be extended to multiple photon states.
\end{abstract}

\maketitle

\section{Introduction}
We consider a model for the propagation of a one-photon state in a random medium composed of two-level atoms. As shown in \cite{Kraisler_22}, the probability amplitudes $\psi(\bx,t)$ and $\phi(\bx,t)$ for exciting an atom and creating a photon, respectively, at the point $\bx$ at the time $t$ obey the equations
\begin{align} 
i\partial_t\phi & = c(-\Delta)^{1/2}\phi + \coup\sqrt{\rho(\bx)} \psi 
\label{eq:continuum1i} \ , 
\\
i\partial_t \psi & = \coup\sqrt{\rho(\bx)}\phi + \Omega \psi 
\label{eq:continuum1ii}\ ,
\end{align}
where $\rho$ is the number density of the atoms, $\Omega$ is their resonance frequency, and $g$ is the atom-field coupling constant. The amplitudes obey the normalization condition
\begin{align}
\int  \left( |\phi(\bx,t)|^2 + \vert\psi(\bx,t)\vert^2 \right)d\bx = 1 \ ,
\end{align}
which has the interpretation that $\vert\phi\vert^2$ is the one-photon probability density and that $\vert\psi\vert^2$ is the atomic probability density. The density $\rho$ is taken to be a random field with prescribed correlations.

In this paper we consider the discrete analog of \eqref{eq:continuum1i} and \eqref{eq:continuum1ii} formulated on the $d$-dimensional lattice $\Z^d$:
\begin{equation} \partial_t \begin{pmatrix} \phi \\ \psi \end{pmatrix} \ = \ H \begin{pmatrix} \phi \\ \psi \end{pmatrix} \ , \quad \phi, \psi \in \ell^2(\Z^d) \ , 
\end{equation}
where the Hamiltonian $H:\ell^2(\Z^d;\C^2)  \to \ell^2(\Z^d;\C^2)$ is a bounded self-adjoint operator of the following form
\begin{equation} \label{eq:Hamiltonian}
    H \begin{pmatrix} \phi \\ \psi \end{pmatrix} \ = \ \begin{pmatrix} T & g\sqrt{\rho} \\ g\sqrt{\rho} & \Omega \end{pmatrix} \begin{pmatrix} \phi\\ \psi \end{pmatrix} \ ,
\end{equation}
where
\begin{enumerate}
\item $T$ is a bounded operator on $\ell^2(\Z^d)$. 

\item $\rho$ is a multiplication operator on $\ell^2(\Z^d)$ with diagonal matrix elements 
\begin{align}
\rho(x)=\rho_0(1+V({x})) \ ,
\end{align}
where the potential $V$ has zero mean, corresponding to density fluctuations about a constant background density $\rho_0$.

\item $g, \Omega >0$ are constant. 
\end{enumerate}

We will refer to the operator $T$ as the hopping operator. Two examples of hopping operators to keep in mind are $T=-\Delta$, the nearest neighbor Laplacian (see eq.\ \eqref{eq:laplacian} below), and $T=(-\Delta)^{1/2}$, which correspond to the cases of polaritons or photons, respectively. However, the analysis extends to a broader class of hopping operators, as described in Section \ref{sec:hamiltonian}. We assume that $\{V(x)\}_{x\in\Z^d}$ are independent and identically distributed (i.i.d.) random variables that are uniformly distributed on $[-1,1]$. Thus $\{\rho(x)\}_{x\in \Z^d}$ are i.i.d. and uniform in $[0,2\rho_0]$.  

Our study of the above problem is motivated by the fact that the analysis of Anderson localization is better developed in discrete rather than continuous settings.  In particular, existing arguments for localization in the continuum rely on the strict locality of the operator $-\Delta$ to obtain the required geometric resolvent identities. To study the spectrum of the operator \eqref{eq:Hamiltonian}, we reduce the system to a single equation of Schr\"odinger type with an energy-dependent potential, which has an energy dependence that diverges near the resonant frequency of the atoms. The main novelty in the methods presented in this work lies in the arguments employed to relate properties of \eqref{eq:Hamiltonian} to those of the reduced model.
We find that there is localization of single photons in an energy band centered at the resonant energy of the atoms. Moreover, if either the coupling constant $g$ or the average density of atoms $\rho_0$ is increased sufficiently, the spectrum is pure-point. These results form the content of the following theorem and corollary.

\begin{theorem}
\label{maintheorem}
Consider the Hamiltonian $H$ as in \eqref{eq:Hamiltonian} with $\coup^2\rho_0 > 0$. There exists $M>0$, independent of $g,\rho_0$ and $\Omega$, such that for almost every realization of the random variables $\{\rho(x)\}_{x\in\Z^d}$, the operator $H$ has only pure-point spectrum in the range $\vert E-\Omega \vert < M\coup^2\rho_0$.
\end{theorem}
We note that $$\left \| \left ( \begin{smallmatrix} 0 & \coup \sqrt{\rho} \\
\coup \sqrt{\rho} & 0 \end{smallmatrix} \right )\right \| \le g\sqrt{2\rho_0} \quad , $$ with equality holding almost surely.  It follows that the spectrum of $H$ is almost surely contained in the union of two intervals:
\begin{equation}\label{eq:Sg}
    \sigma(H) \ \subset \ [-\|T\|-\coup \sqrt{2\rho_0},\|T\|+\coup \sqrt{2\rho_0} ]  \bigcup [\Omega-\coup \sqrt{2\rho_0}, \Omega +\coup \sqrt{2\rho_0} ]  \ .
\end{equation}
In particular if $\coup^2 \rho_0$ is sufficiently large, then the interval of pure-point spectrum guaranteed by Theorem \ref{maintheorem} covers the entire spectrum of $H$, leading to the following corollary.

\begin{cor}\label{cor:fullypp}
There exists a constant $C>0$ such that if $g^2\rho_0 > C$, the spectrum of $H$ is almost surely pure-point. Specifically, for a.e. realization of the random variables  $\{\rho\}_{x\in \Z^d}$ we have $\sigma_c(H)=\emptyset$.
\end{cor}
\begin{remark}
Here $\sigma_c(H)$ denotes the {continuous spectrum} of $H$ (see Definition \ref{def:ctsandpointspectrum} below).
\end{remark}

This paper is organized as follows. In \Cref{sec:hamiltonian}, we define the Hamiltonian \eqref{eq:Hamiltonian}, which is a discrete analog of the system \eqref{eq:continuum1i}-\eqref{eq:continuum1ii}. In \Cref{sec:Background}, we recall the necessary definitions and criteria under which an operator exhibits spectral and dynamical localization. In \Cref{sec:finitereduction}, we reduce the problem to one involving a finite dimensional operator and sketch the proof of \Cref{maintheorem}. In \Cref{sec:fractionalmoments}, we derive the necessary bounds to show that this operator exhibits dynamical localization, while in \Cref{sec:correlatorbound} we combine these bounds to prove the main theorem. We conclude the paper the extension of our results to the multi-particle problem for bosonic atoms (\Cref{sec:bosonic}).

\section{One Excitation Hamiltonian}\label{sec:hamiltonian}

The Hilbert space of interest for the analysis of the discrete analog of \eqref{eq:continuum1i} and \eqref{eq:continuum1ii} is
\begin{align}
\label{hilbert}
\mathcal H =   \ell^2(\Z^d;\C^2) \cong \ell^2(\Z^d)\otimes \C^2\cong  \left \{\begin{pmatrix}\phi \\ \psi \end{pmatrix} : \phi,\psi\in\ell^2(\Z^d;\C) \right \}\ ,
\end{align}
with the inner product 
\begin{align}
    \left \langle \begin{pmatrix}
        \phi_1 \\ \psi_1
    \end{pmatrix} \, , \  \begin{pmatrix}
        \phi_2 \\ \psi_2
    \end{pmatrix} \right \rangle = \sum_{x\in\Z^d}\left( \phi_1^*(x)\phi_2(x) + \psi_1^*(x)\psi_2(x)\right).
\end{align}
Let $\delta_x\in\ell^2(\Z^d)$ denote the basis element defined by
\begin{align}
    \delta_x(y) = \begin{cases}
        1 & \text{if } x = y \ , \\
        0 & \text{if } x\neq y \ .
    \end{cases}
\end{align}
We will use the basis of $\mathcal{H}$, which is comprised of functions of the form
\begin{equation} e_{1,x} = \begin{pmatrix}
    \delta_x \\ 0 
\end{pmatrix} \ , \quad e_{2,x} \begin{pmatrix}
    0 \\ \delta_x 
\end{pmatrix} \ . 
\end{equation}
We consider the Hamiltonian $H:\mathcal H \to \mathcal H$ given by \eqref{eq:Hamiltonian}, and write $H=H_0 + \mathcal{P}$, where
\begin{align}
     H_0 \begin{pmatrix}\phi \\ \psi\end{pmatrix} \ &= \ \begin{pmatrix} T\phi\\ \Omega \psi \end{pmatrix}\label{FreeHamiltonian}\ , \\
   \mathcal{P}\begin{pmatrix}\phi \\ \psi\end{pmatrix} \ &= \ g\begin{pmatrix}\sqrt{\rho}\psi \\ \sqrt{\rho}\phi \end{pmatrix}  \ .
    \label{Random}
\end{align}
for any $\phi,\psi\in \ell^2(\Z^d)$.
Here $H_0$ is the \emph{free Hamiltonian} and $\{\rho(x)\}_{x\in\Z^d}$ are a collection of i.i.d random variables uniformly distributed in the interval $[0,2\rho_0]$, for a positive constant $\rho_0$.
The operator $T:\ell^2(\Z^d)\to\ell^2(\Z^d)$ is self-adjoint and is defined by its matrix elements $T(x,y)$: 
\begin{align}
    T\phi(x) \ = \ \sum_{y\in\Z^d}T(x,y)\phi(y)\, , \quad \phi\in\ell^2(\Z^d) ,
\end{align}
which we require to satisfy
\begin{align}\label{eq:summability}
    \sup_{x}\sum_{y}\vert T(x,y)\vert^s \ < \ \infty\, ,
\end{align}
for some $0< s< 1$. 

The discrete Laplacian $-\Delta:\ell^2(\Z^d)\to\ell^2(\Z^d)$ is defined as
\begin{align}\label{eq:laplacian}
    (-\Delta\phi)(x) \ = \ \sum_{\|x-y\|_{1}=1}(\phi(x)-\phi(y))\ , 
\end{align}
which has matrix elements
\begin{align}
    -\Delta(x,y)\ = \ \begin{cases}
        -1 & \text{if }\|x-y\|_1 = 1 \ , \\ 
        2d & \text{if } x = y \ , \\
        0 & \text{otherwise} \ .
    \end{cases}  \ , 
\end{align}
which clearly satisfy the condition \eqref{eq:summability}, by virtue of the sum being finite and independent of $y\in\Z^d$.
Introducing the Fourier transform
$\mathcal{F}:\ell^2(\Z^d)\to L^2([0,2\pi]^d)$ and its inverse\\ $\mathcal{F}^{-1}:L^2([0,2\pi]^d)\to\ell^2(\Z^d)$ by
\begin{align}
     (\mathcal{F}\phi)(k) \ &= \ \sum_{x\in\Z^d}\phi(x)e^{-ik\cdot x}\ ,\\
     (\mathcal{F}^{-1}f)(x)\ &=\ \frac{1}{(2\pi)^{d}}\int_{[0,2\pi]^d} e^{ik\cdot x} f(k) dk \  ,
\end{align}
it is a straightforward calculation to see that
\begin{align}\label{eq:Deltamultiplier}
    \left(-\Delta\phi\right)(x) = \mathcal{F}^{-1}\left( h\mathcal{F}\phi\right)(x)\ ,
\end{align}
where the function $h(k)$ is given by
\begin{align}
     h(k) \ = \ 4\sum_{i=1}^{d}\sin^2\left(\tfrac{k_i}{2}\right)\ .
\end{align}
The Fourier representation \eqref{eq:Deltamultiplier} of $-\Delta$ leads to the definition of $(-\Delta)^{1/2}:\ell^2(\Z^d)\to\ell^2(\Z^d)$ through the formula
\begin{align}
\left((-\Delta)^{1/2}\phi\right)(x) \ = \ \mathcal{F}^{-1}\left( h^{1/2}\mathcal{F}\phi\right)(x)\, .
\end{align}
The matrix elements of $(-\Delta)^{1/2}$ satisfy the condition \eqref{eq:summability} for any $s > d/(d+1)$, as one can see using the estimate
\begin{align}
    \vert(-\Delta)^{1/2}(x,y)\vert \ \leq \ \frac{c_d}{\vert x-y\vert^{d+1}}\, ,
\end{align}
proved in \cite[Theorem 2.2]{Gebert_20}.

\section{Background and Localization Criteria}\label{sec:Background}
In this section we define spectral and dynamical localization for self-adjoint operators on $\mathcal{H}$ and state a key criterion for dynamical localization in terms of \emph{eigenfunction correlators}. We then define strong dynamical localization for random Hamiltonians.

First we recall the notions of {continuous} and {point} spectrum. 
\begin{defn}[Continuous and Point Spectrum]\label{def:ctsandpointspectrum} Let $K$ be a self-adjoint operator on a Hilbert space $\mathcal{K}$. 
\begin{enumerate}
    \item The \textbf{continuous subspace} for $K$, denoted $\mathcal{K}^c$, is the subspace consisting of vectors with purely continuous spectral measure, and the \textbf{continuous spectrum} of $K$, denoted $\sigma_c(K)$, is the spectrum of the restriction of $K$ to $\mathcal{K}^c$. 
    \item The \textbf{pure-point subspace} for $K$, denoted $\mathcal{K}^{pp}$, is the closure of the span of the eigenvectors of $K$ and the \textbf{point-spectrum}, denoted $\sigma_p(K)$, is the closure of the set of eigenvalues of $K$, i.e., the spectrum of the restriction of $K$ to $\mathcal{K}^{pp}$.
\end{enumerate}
\end{defn}
The so-called {RAGE theorem} is a fundamental result which gives dynamical descriptions of the continuous and pure-point subspaces.
\begin{theorem}[RAGE]
Let $K$ be a self-adjoint operator on a Hilbert
space $\mathcal{K}$ and let $\{A_n\}_{n=1}^\infty$ be a sequence of compact operators on $\mathcal{K}$ which converges
strongly to the identity. Then
\begin{align}
    \mathcal{K}^c \ &= \ \left\{\psi\in\mathcal{K}\ \middle |\ \lim_{n\to\infty}\lim_{T\to\infty}\frac{1}{T}\int_0^T\|A_n e^{-iKt}\psi\|dt = 0\right\}\, \\ \intertext{and}
    \mathcal{K}^{pp} \ &= \ \left\{ \psi\in\mathcal{K}\ \middle | \ \lim_{n\to\infty}\sup_{t\in\R}\|(1-A_n)e^{-iKt}\psi\| = 0\right\}\, .
\end{align}
\end{theorem}
\begin{remark} The above theorem summarizes separate results due to Ruelle, Amrein and Georgescu and Enss.  See \cite[Theorem 2.6]{aizenman2015random} for the version stated here and references to the original literature.
\end{remark}
Now we define two notions of localization for self-adjoint operators acting on the Hilbert space $\mathcal{H}=\ell^2(\Z^d;\C^2)$. For each $x$, we let $\pi_x:\mathcal{H}\to \mathbb{C}^2$ denote evaluation at $x$:
\begin{equation}
    \pi_x\begin{pmatrix}
        \phi_1 \\
        \phi_2
    \end{pmatrix} \ : = \ \begin{pmatrix} \phi_1(x) \\ \phi_2(x)
    \end{pmatrix} \ .
\end{equation}
\begin{defn}
Let $H$ be a self-adjoint operator on $\mathcal{H}$ and fix an interval $I\subset\R$. 
\begin{enumerate}
    \item We say that $H$ exhibits \textbf{spectral localization} in the interval $I$ if $$\sigma_c(H)\cap I\ =\ \emptyset \ . $$ 
    \item We say that $H$ exhibits \textbf{dynamical localization} in the interval $I$ if for each $x\in \Z^d$
    \begin{align}
    \sum_{y\in \Z^d} \sup_{t\in\R} \| \pi_y P(I) e^{-itH}\pi_x^\dagger \|^2 \ < \ \infty \ ,
    \end{align}
    where $P(I)$ denotes the spectral projection for $H$ associated to the interval $I$ and $\pi^{\dagger}_x$ denotes the adjoint of $\pi_{x}$. 
\end{enumerate}
\end{defn}

Given a non-negative function $W(E)$ on $\R$, we introduce the {eigenfunction correlator} $Q(y,x;W)$ associated to $H$ with weight $W$:
\begin{align}\label{eq:QW}
    Q(y,x;W) \ = \ \sup_{F\in C_0(\R),\ \|F\|_{\infty}\leq 1} \| \pi_y, W(H) F(H)\pi_{x}^\dagger \| \ .
\end{align}
It is clear that
\begin{align}
    \sup_{t\in\R} \| \pi_{y}P(I)e^{-itH}\pi_{x}^{\dagger}\| \ \leq \  Q(y,x;\chi_I)\ .
\end{align}
Thus for dynamical localization to hold in $I$, it is sufficient that $\sum_y Q(y,x;\chi_I)^2 <\infty$ for each $x$. More generally, we say that $H$ satisfies \textbf{weight $W$ dynamical localization} if
\begin{align}
    \sum_{y\in\Z^d}Q(y,x;W)^2 \ < \ \infty \ , \quad  \text{for each }x\in \Z^d \ .
\end{align}
In the proof of our main theorem, we use  the weight $$W(E)=\min \left \{\left( \frac{|E-\Omega|}{g^{2}\rho_0} \right )^\eps, \chi_{I}(E) \right \} \ , $$ with $\eps>0$ arbitrary.
The RAGE theorem can be used to show that weight $W$ dynamical localization implies spectral localization on the set $\{ W >0\}$.
\begin{lemma}\label{lem:dl->sl}
Let $H$ be a self-adjoint operator on $\mathcal{H}$ and let $W\ge 0$ be a piecewise continuous function. If $H$ exhibits weight $W$ dynamical localization, then $H$ exhibits spectral localization on $\{E \ : \ W(E) >0\}$.
\end{lemma}
\begin{proof}
It suffices to show that $P_c(K)=0$ for any compact set with $K\subset \{E : W(E) >0\} $. Let $K$ be such a compact set. By the RAGE theorem, if $\psi\in\mathcal{H}$, then
\begin{align}
    \|P_c(K)\psi\|^2 \ = \ \lim_{L\to\infty}\lim_{T\to\infty}\frac{1}{T}\int_0^T\| \chi_{\Lambda_L^c}e^{-itH}P(K) \psi\|^2dt\ ,
\end{align}
where $\Lambda_L = [-L,L]^d$.  It suffices to prove that the above limit vanishes for functions of finite support.  Suppose that $\psi = \sum_{x\in V} \pi^\dagger_x \psi(x)$ with $V$ finite. Then
\begin{equation}
\begin{split}
    \| \chi_{\Lambda_L^c}e^{-itH}P(K)\psi \|^2\ & =\sum_{x,x'\in V} \sum_{y\notin \Lambda_L} \vert \langle 
          \pi_y e^{-i t H} \pi_x^\dagger \psi(x) , \pi_y e^{-i tH}\pi_{x'}^\dagger \psi(x') \rangle \vert  \\
          & \le \ \sum_{x,x'\in V} \| \psi(x) \|^2 \sum_{y\notin \Lambda_L} \| \pi_y e^{-itH}P(K) \pi_x^\dagger \|^2  \\
    &\leq \ \frac{1}{\delta(K)^2} |V| \sum_{x\in V} \sum_{y\notin \Lambda_L} Q(y,x;W)^2 \ ,
\end{split}
\end{equation}
where $\delta(K)=\min_{E\in K} W(E)$ and in the first estimate we have used the Cauchy-Schwarz inequality.
Since $H$ exhibits weight $W$ dynamical localization, the right hand side goes to $0$ as $L\to\infty$.
\end{proof}

In the following, we will be concerned with \emph{random} operators on $\mathcal{H}$. Let $(\Omega,\Prob)$ be a probability space and let $\E(\cdot)=\int_\Omega \cdot \ d \Prob$. A \textbf{random Hamiltonian} on $\mathcal{H}$ is a measurable map $\omega \mapsto H(\omega)$ from $\Omega$ into the bounded self-adjoint operators on $\mathcal{H}$. We follow the usual probabilistic notation of suppressing the independent variable $\omega$ in many expressions, unless needed for clarity. 
\begin{defn}
Let $H=H(\omega)$ be a random Hamiltonian on $\mathcal{H}$ and let $I\subset \R$ be an interval. We say that
\begin{enumerate}
    \item $H$ exhibits \textbf{spectral localization} in $I$ if $\sigma_c(H)=\emptyset$ almost surely.
    \item $H$ exhibits \textbf{strong dynamical localization} in $I$ if 
     \begin{align}
       \sum_{y\in \Z^d} \E\left[\sup_{t\in\R} \| \pi_y P(I) e^{-itH}\pi_x^\dagger \|^2\right] \ < \ \infty \ .
    \end{align}
    \item $H$ exhibits \textbf{weight $W$ strong dynamical localization} if
    \begin{align}\label{eq:SDLcondition}
        \sup_{x\in \Z^d} \sum_{y\in\Z^d} \E \left [Q(y,x;W)^2 \right ] \ < \ \infty \ ,
    \end{align}
    where $Q(x,y;W)$ is the random eigenfunction correlator for $H$, defined in \eqref{eq:QW}.
\end{enumerate}    
\end{defn}

If $H$ exhibits weight $W$ strong dynamical localization, it follows that $H(\omega)$ exhibits weight $W$ dynamical localization for $\Prob$-almost every $\omega$, and thus by Lemma \ref{lem:dl->sl}, $H$ exhibits spectral localization on $\{ W>0\}$.

\section{Finite Volume Reduction and an Outline of the proof of the main theorem}\label{sec:finitereduction}
In the proof of \Cref{maintheorem}, we will make use of the restriction of the random Hamiltonian $H$ to finite volumes.  We introduce the notation
$\Lambda \ssubset \Z^d$ if and only if $\Lambda$ is a finite subset of $\Z^d$.
Given $\Lambda \ssubset \Z^d$, let $\mathcal{H}_\Lambda=\ell^2(\Lambda;\C^2)$ and let $\pi_\Lambda:\mathcal{H}\to \mathcal{H}_\Lambda$ be the natural restriction map.

The adjoint $\pi_\Lambda^{\dagger}:\mathcal{H}_{\Lambda}\to\mathcal{H}$ is then the injection map which extends a function to be zero on $\Lambda^c$.
Given an operator (or a random operator) $M$ on $\mathcal{H}$ we define
\begin{align}
    M_\Lambda \ = \ \pi_\Lambda M \pi_\Lambda^{\dagger} \ . 
\end{align}
Since $M_\Lambda$ is an operator on a finite dimensional space, it has only point spectrum.

The eigenvalue problem for the restriction $H_\Lambda$ of the random Hamiltonian $H$ to $\mathcal{H}_\Lambda$ can be reduced to a single equation for $\phi\in\ell^2(\Lambda)$, which is nonlinear in the eigenvalue parameter $E$. Written in terms of the components $\phi,\psi\in \ell^2(\Lambda)$, the eigenvalue problem for $H_\Lambda$ is of the form
\begin{align}
    T_{\Lambda}\phi + g \sqrt{\rho_\Lambda}\psi &= E\phi\ , \\
    g \sqrt{\rho_\Lambda}\phi + \Omega\psi &= E\psi\ .
\end{align}
Provided $E\neq \Omega$, the above may be reduced to the single equation
\begin{align}
    \left(T_{\Lambda}-E+\tfrac{g^2}{E-\Omega} \rho \right)\phi = 0\ .
\end{align}
For $E\in \mathbb{R}\setminus\{ \Omega\}$ fixed, we define $K_\Lambda(E): \ell^2(\Lambda)\to \ell^2(\Lambda)$ by  
\begin{align}\label{eq:Reducedfvhamiltonian}
    K_\Lambda(E)  \ = \ T_\Lambda + \tfrac{g^2}{E-\Omega} \rho_\Lambda \ .
\end{align}
Note that $E\neq \Omega$ is an eigenvalue of $H_{\Lambda}$ if and only if $E$ is an eigenvalue of $K_\Lambda(E)$. 

We now briefly outline the main steps of the proof of \Cref{maintheorem}. In \Cref{sec:fractionalmoments}, adapting standard arguments on random Schr\"odinger operators, we obtain bounds on the fractional moments of the finite volume Green's functions \begin{align}
G_{\Lambda}(y,x,E) \ = \ \langle \delta_y, (K_\Lambda(E)-E)^{-1} \delta_x \rangle    
\end{align}
of the form
\begin{align}\label{eq:gfbound}
    \sum_{y\in\Z^d}\E[\vert G_\Lambda(y,x,E)\vert^s]\ \leq \ C(s,E)\left(\frac{\vert E-\Omega\vert}{g^2\rho_0}\right)^s\ ,
\end{align}
for a constant $C(s,E)$. Next, in \Cref{sec:correlatorbound}, we use the bounds from the previous section to prove the main theorem as follows. We note that the eigenfunction correlator $Q(x,y,W)$, for the infinite volume operator $H$, can be bounded in terms  of a limit of finite volume correlators: 
\begin{align}
    Q(y,x;W) \ \le \ \liminf_{\Lambda \uparrow \Z^d} Q_\Lambda(y,x;W)  \ ,
\end{align}
where the eigenfunction correlator for $H_\Lambda$ can be expressed as 
\begin{align}\label{eq:QLambda}
    Q_\Lambda(y,x;W) \ = \ \sum_{E\in\sigma(H_\Lambda)} W(E) \|\pi_y  P(\{E\}) \pi_x^\dagger \| \ ,
\end{align}
where $P(\{E\})$ is the projection onto the spectral subspace of eigenvectors of the operator $H_\Lambda$ with eigenvalue $E$.\footnote{If the eigenvalue $E$ is non-degenerate, then 
$$ \pi_y P(\{E\}) \pi_x^\dagger  \ = \ \begin{pmatrix} \phi_{E,1}(y) \phi_{E,1}(x) & \phi_{E,1}(y) \phi_{E,2}(x) \\
\phi_{E,2}(y) \phi_{E,1}(x)  &\phi_{E,2}(y) \phi_{E,2}(x)
\end{pmatrix} \ , $$
where $\phi_E^T = (\phi_{E,1}, \phi_{E,2})$ is the corresponding normalized eigenvector.  The expression  $\pi_y P(\{E\}) \pi_x^\dagger $ used in \eqref{eq:QLambda} allows for possible degeneracy of the eigenvalues.
} 
We then use \eqref{eq:gfbound} to show that
\begin{align}
    \sup_x \sum_y \E \left [ Q(y,x,W)^2 \right ] \ \le \ \sup_x \liminf_{\Lambda \uparrow \Z^d}  \sum_{y} \E \left [ Q_\Lambda(y,x;W)^2 \right ]  \ < \ \infty \ ,
\end{align}
which corresponds to \eqref{eq:SDLcondition}, which is a sufficient condition for strong dynamical localization.

\section{Green's Function and Fractional Moment Bounds} \label{sec:fractionalmoments}
Let $\Lambda \subset \Z^d$ be a finite set. We define the \textbf{finite volume reduced Green's function} to be the rational function 
\begin{equation}
   G_\Lambda(x,x';z) \ = \ \langle e_{1,x}, (H_\Lambda-z)^{-1} e_{1,x'} \rangle \ , \quad x,x' \in \Lambda \  .
\end{equation} 
It follows that that $G_\Lambda(x,x';z)\in \R\cup \{\infty\}$ for $z\in \R$, with  poles contained in $\sigma(H)$. For $z\not \in \sigma(H)\cup \{\Omega\}$, $G_\Lambda$ solves the equation
\begin{align} \label{eq:Greensfunction}
    \left(K_\Lambda(z)-z\right)G_\Lambda(x,x';z) = \delta_{x,x'} \ .
\end{align}
Here we recall that $z \in \sigma(H)$ if and only if $z \in \sigma(K_\Lambda (z))$.  

Following Ref.~\cite{aizenman_93}, we may obtain an \textit{a priori} bound on suitable averages of the Green's function.  By the Schur complement formula, 
\begin{align}\label{eq:diagresolvent}
    G_{\Lambda}(x,x;z) \ = \ \frac{1}{\frac{g^2}{z-\Omega} \rho(x) - z - \gamma_\Lambda(x;z)} \ ,
\end{align}
where 
\begin{align}
\gamma_\Lambda(x;z) \ = \ \sum_{y,y'\in \Lambda\setminus \{x\}} T(x,y) T(y',x) G_{\Lambda\setminus \{x\}} (y,y',z)
\end{align}
is independent of $\rho(x)$.  It follows that 
\begin{multline}\label{eq:apriori}
    \E[\vert G_{\Lambda}(x,x;z)\vert^s | \{\rho(y) : y\neq x\}] \\ \leq \ \sup_{\alpha \in\R} \left (\frac{|z-\Omega|}{g^{2}} \right )^{s} \frac{1}{2\rho_0}\int_0^{2\rho_0} \frac{1}{|\rho - \alpha|^s} \ \le \  \frac{1}{1-s} \left(\frac{\vert z-\Omega\vert}{ g^{2}\rho_0}\right)^s\ .
\end{multline}
Similarly, for $x\neq x'$, we have
\begin{equation}\label{eq:offdiagresolvent}
     \E[\vert G_{\Lambda}(x,x',z)\vert^s | \{\rho(y) : y\notin \{x,x'\} \}] \ \le \ \frac{4^s}{1-s} \left(\frac{\vert z-\Omega\vert}{ g^{2}\rho_0}\right)^s \ .
\end{equation}
See \cite[Theorem 8.3]{aizenman2015random}.

Applying the resolvent identity to $H_\Lambda$ and $H_{\Lambda\setminus\{x\}}$, with $x\in \Lambda$, we find that 
\begin{align}
    G_{\Lambda}(x,x',z) = G_{\{x\}}(x,x,z)\delta_{x,x'} - \sum_{y\neq x}G_{\Lambda}(x,x,z)T(x,y)G_{\Lambda\setminus \{x\}}(y,x',z)\ .
\end{align}
For $x\neq x'$, we obtain the one-step bound
\begin{align}
    \E[\vert G_{\Lambda}(x,x',z)\vert^s]\leq \sum_{y\neq x}\vert T(x,y)\vert^s\E\left[\vert G_{\Lambda}(x,x,z)\vert^s \vert G_{\Lambda\setminus \{x\}}(y,x',z)\vert^s\right] \ .
\end{align}
Note that $G_{\Lambda\setminus\{x\}}$ is independent of the random variable $\rho(x)$.   Computing the expectation by first conditioning on $\rho(x')$ for  $x'\neq x$ and using \eqref{eq:apriori} to bound the average over $\rho(x)$, we obtain
\begin{multline}
    \label{eq:bound1}
     \E[\vert G_{\Lambda}(x,x';z)\vert^s] \\ \leq \ \frac{1}{1-s}\left(\frac{\vert z-\Omega\vert}{ g^{2}\rho_0}\right)^s\begin{cases} 1 & x=x' \\
     \sum_{y\neq x}\vert T(x,y)\vert^s\E\left[ \vert G_{\Lambda\setminus \{x\}}(y,x';z)\vert^s\right] & x\neq x'\end{cases}\ .
\end{multline}

We now define $\lambda(T)$ and $r_{z,s}$ as
\begin{align}
    \lambda(T)&= \inf_{s\in(0,1)}\left[\frac{1}{1-s}\sup_{x}\sum_{y\neq x} \vert T(x,y)\vert^s\right]^{1/s}\ , \\
    r_{z,s} &= \frac{1}{1-s}\left(\frac{\vert z-\Omega\vert}{g^2\rho_0}\right)^s\sup_{x}\sum_{y\neq x}\vert T(x,y)\vert^s\ .
\end{align}
Note that $\lambda(T) < \infty$ by the assumption \eqref{eq:summability}.

Furthermore there exists an $s\in(0,1)$ such that $r_{z,s}< 1$ if and only if $\frac{g^2\rho_0}{\vert z-\Omega\vert} >  \lambda(T)\  $,
or equivalently $ \vert z-\Omega\vert  <  \frac{g^2\rho_0}{\lambda(T)}$.

The next proposition allows us to transform the single step bound \eqref{eq:bound1} into a global bound on the fractional moments of the finite volume Green's functions. We point out that this result is essentially an application of \cite[Theorem 9.2]{aizenman2015random}; however, to keep our presentation self-contained, we include a proof.
\begin{prop}\label{Greensfunctionbound}For any $z\in\R$ such that that $\vert z-\Omega\vert < \frac{g^2\rho_0}{\lambda(T)}$, there exists $0<s<1$, such that for all $x_0\in\Z^d$,
\begin{align}\label{eq:propequation}
    \sum_{x\in\Z^d}\sup_{\Lambda\ssubset\Z^d}\E[\vert G_\Lambda(x,x_0;z)\vert^s]\leq \frac{1}{(1-s)(1-r_{z,s})}\left(\frac{\vert z-\Omega\vert}{g^2\rho_0}\right)^s \ .
\end{align}
\end{prop}
\begin{remark}
Here and below we use the convention that $G_\Lambda(x,x_0;z) = 0$ if $\{x,x_0\} \not \subset \Lambda$.
\end{remark}
\begin{proof}
Fix $x_0\in\Lambda$ and make the identifications 
\begin{align}
    f(x) \ &= \ \sup_{\Lambda\ssubset\Z^d}\E[\vert G_{\Lambda}(x,x_0;z)\vert^s]\ , \\
    K(x,u) \ & = \ \frac{1}{1-s}\left(\frac{\vert z-\Omega\vert }{g^2\rho_0}\right)^s\vert T(x,u)\vert^s\ , \\
    g(x) \ & = \  \frac{1}{1-s} \left(\frac{\vert z-\Omega\vert}{g^2\rho_0}\right)^s\delta_{x,x_0}\ .
\end{align}
By \eqref{eq:diagresolvent} and \eqref{eq:offdiagresolvent}, $f(x)$ is finite and uniformly bounded. Clearly $K(x,u)\geq 0$ and we have
\begin{align*}
    r_{z,s} = \ \sup_{x}\sum_{u}K(x,u) \ <  \ 1\ ,
\end{align*}
by the definition of $\lambda(T)$ and our assumption on $z$. We also have
\begin{align*}
    f(x) \ \leq \ g(x) + (Kf)(x)\ ,
\end{align*}
by \eqref{eq:bound1}. Iterating this bound, we find that
\begin{align}
    f(x) \ \leq \ \sum_{n=0}^{N-1} (K^n g)(x) + (K^Nf)(x)\, ,
\end{align}
for each $N\in \N$.
Since $\|K\|_{\infty,\infty} <1$ and $f\in \ell^\infty$, the remainder term tends to $0$ and
\begin{align}
\begin{split}
    \sum_{x} f(x) \ \leq \ \sum_{x}\sum_{n=0}^{\infty} (K^n g)(x)\ \leq \ \sum_{x}\sum_{n=0}^{\infty}r_{z,s}^n g(x)\
    = \ \frac{1}{1-s} \left(\frac{\vert z-\Omega\vert}{g^2\rho_0}\right)^s\frac{1}{1-r_{z,s}}\, , 
\end{split}
\end{align}
which is the desired inequality.
\end{proof}
\section{Correlator Bound} \label{sec:correlatorbound}
It remains to bound the eigenfunction correlator by the fractional moment of the Green's function. Consider a finite volume $\Lambda \subset \Z^d$, and denote $H_\Lambda$, $G_\Lambda$ and $K_\Lambda(E)$ by $H$, $G$ and $K(E)$ for simplicity. At the end of the argument, we will take $\Lambda \uparrow \Z^d$ to obtain bounds in the infinite volume.  For now we work in a fixed finite volume.

The arguments used here follow those in \cite[Chapter 7]{aizenman2015random} quite closely, however with some key alterations to adapt the argument to the energy-dependent potential. The first observation is that the matrix form of the eigen-projections takes a simple form:
\begin{equation}\label{eq:PEmatrix}
   \pi_y P(\{E\}) \pi_x^\dagger
        \ = \ \begin{pmatrix} 1 & \frac{g\sqrt{\rho(x)}}{E-\Omega} \\
         \frac{g\sqrt{\rho(y)}}{E-\Omega} & \frac{g^2\sqrt{\rho(y)\rho(x)}}{(E-\Omega)^2} 
    \end{pmatrix} \langle e_{1,y}, P(\{E\}) e_{1,x} \rangle \ .
\end{equation}
This can be seen by using the equation $(H-E)P(\{E\})=0$ to solve for the second row of the matrix in terms of the first, and then using the equation $P(\{E\})(H-E)$ to do the same with the columns. We use \eqref{eq:PEmatrix} to bound the weighted correlator:
\begin{align}\label{eq:QLambdaub}
    Q_\Lambda(y,x;W) \ \leq  \ \sum_{E\in\sigma(H) }  W_g(E) \vert \langle e_{1,y}, P(\{E\}) e_{1,x} \rangle\vert \ =: \ \tilde Q_\Lambda(y,x;W) ,
\end{align}
where
\begin{align}
\label{eq:WE}
W_g(E) \ & = \  \ W(E)\left ( 1 + \frac{g^2 \rho_0}{(E-\Omega)^2} \right ) \ . \intertext{We take }
    \label{eq:W}
    W(E) \ & = \ \min\left\{ \left (\frac{|E-\Omega|}{g^2\rho_0} \right )^\eps,\chi_I(E) \right\} \ ,
\end{align}
where $I=(\Omega-\frac{g^2\rho_0}{\lambda(T)}, \Omega + \frac{g^2\rho_0}{\lambda(T)})$, as in \Cref{Greensfunctionbound} with $\eps>0$ arbitrary.

Our argument focuses on the correlator $\tilde Q_\Lambda(\cdot,\cdot,W)$ defined in \eqref{eq:QLambdaub}. To streamline the notation, we set
\begin{equation}\label{eq:phiE}
\phi_E(y,x) \ = \ \langle e_{1,y}, P(\{E\}) e_{1,x} \rangle
\end{equation}
for $E\in \sigma(H)$.  We note that $\phi_E(\cdot,x)$ is an eigenfunction of $K_\Lambda(E)$ (defined in \eqref{eq:Reducedfvhamiltonian}) with eigenvalue $E$, and furthermore that
\begin{equation}\label{eq:phiEnorm}
\begin{split}
    \sum_{y\in \Z^d} \left (1 + \frac{g^2 \rho(y)}{(E-\Omega)^2} \right ) |\phi_E(y,x)|^2 \ &= \ \sum_{y\in \Z^d} \tr \pi_y P(\{E\}) |e_{1,x}\rangle \langle{e_{1,x}}| P(\{E\}) \pi_y^\dagger \\ &= \ \langle e_{1,x}, P(\{E\}) e_{1,x} \rangle \ = \ \phi_E(x,x) \ . 
    \end{split}
\end{equation}
Also, since $W(E)\le 1$,
\begin{equation}
\begin{split}
    \tilde Q_\Lambda(x,x;W) \ &\le \ \sum_{E\in \sigma(H)} \left ( 1 + \frac{g^2 \rho(x)}{(E-\Omega)^2} \right )  \phi_E(x,x)  \\ & = \ \sum_{E\in \sigma(H)} \tr \pi_x P(\{E\}) \pi_x^\dagger \ = \ \tr \pi_x^\dagger \pi_x \ = \ 2 \ 
    \end{split}
\end{equation}
for each $x$.  By the Cauchy-Schwarz inequality applied to the inner product, we have 
\begin{equation}\label{eq:ipCS}
     \vert\phi_E(y,x)\vert \ \le \ \sqrt{ \phi_E(y,y) \phi_E(x,x)} \ .
\end{equation}
A second application of the Cauchy-Schwarz inequality to the summation over $E$ yields
\begin{equation}\label{eq:Qtildebound}
\tilde Q_\Lambda(y,x;W) \ \le \ 2 \ .
\end{equation}

To bound $\E[Q_\Lambda(y,x;W)]$ in terms of the fractional moments of the Green's function, we follow \cite{aizenman2015random} and introduce, for $s\in (0,1)$, the modified correlator
\begin{equation}
    \tilde Q_\Lambda(y,x;W,s) \ = \ \sum_{E\in\sigma(H) }  W_g(E)  \vert \phi_E(y,x)\vert^s \vert\phi_E(x,x)\vert^{1-s} \ .
\end{equation}
We also note the upper bound
\begin{equation}\label{eq:modifiedcorrub}
    \tilde Q_\Lambda(y,x;W) \ \le \ \sqrt{ \tilde Q_\Lambda(x,y;W,s)\tilde Q_\Lambda(y,x;W,s) } \ .
\end{equation}
The proof of \eqref{eq:modifiedcorrub}, which closely follows the proof of \cite[Eq. (7.44)]{aizenman2015random}, involves several applications of the Cauchy-Schwartz inequality, first to the inner product (see \eqref{eq:ipCS}) and then to the summation over $E$.

We now proceed to estimate $\E[\tilde Q_\Lambda(y,x;W,s)]$ in terms of the Green's function.  Let $K(E)$ be the reduced Hamiltonian defined in \eqref{eq:Reducedfvhamiltonian}. Let $x\in \Z^d$ and define $\hat{K}_x(E)$ to be the operator obtained from $K(E)$ by replacing the density $\rho(x)$ with an arbitrarily chosen value $\hat \rho_x$. That is,
\begin{align}
    \hat{K}_x(E) \ = \ K(E) + \frac{g^2\hat \rho_x- g^2\rho(x)}{E-\Omega} \delta_x \langle \delta_x, \cdot \rangle  \ .
\end{align}
If $E\in \sigma(H)$, then $K(E)\phi_E(\cdot,x) = E \phi_E(\cdot,x)$ and thus
\begin{equation}\label{eq:formofphiE}
\begin{split}
\phi_E(y,x) \ &= \ \frac{g^2\hat \rho_x- g^2\rho(x)}{E-\Omega} \langle \delta_y, (\hat{K}_x(E)-E)^{-1}\delta_x\rangle \phi_E(x,x) \\ &= \ \frac{g^2\hat \rho_x- g^2\rho(x)}{E-\Omega}\hat G(y,x,E) \phi_E(x,x)\ ,   
\end{split}
\end{equation}
where $\hat G(y,x,E)$ is the Green's function associated to $\hat K_x(E)$: 
\begin{align}
    \hat G(y,x,E) \ = \ \langle \delta_y, (\hat K_x(E)-E)^{-1}\delta_x\rangle\, .
\end{align}
Upon setting $y=x$ in \eqref{eq:formofphiE}, we see that 
\begin{align}\label{eq:spectrumcharacterization}
\text{$E\in \sigma(H)$ with $\phi_E(x,x)\neq 0$}\quad \text{if and only if} \quad
 \Gamma_x(E)  \ = \ \frac{g^2\rho(x)-g^2\hat\rho_x}{E-\Omega}\ ,
\end{align}
where 
\begin{align}
    \Gamma_x(E) \ = \ -\langle \delta_x,(\hat{K}_x(E)-E)^{-1}\delta_x\rangle^{-1} \ .
\end{align}
This characterization of the spectrum of $H$ leads to the identity
\begin{align}\label{eq:spectmeas}
    \sum_{E\in \sigma(H) } f(E) \phi_E(x,x)   \ = \ \int f(E) \, \delta\left(\Gamma_x(E) - \frac{g^2\rho(x)-g^2\hat\rho_x}{E-\Omega}\right) dE\ ,
\end{align}
which is valid for any function defined on a neighborhood of the spectrum. In obtaining \eqref{eq:spectmeas} we have made use of the identities, 
\begin{equation}
\begin{split}
    \frac{d\Gamma_x(E)}{dE}+\frac{g^2\rho(x)-g^2\hat\rho_x}{(E-\Omega)^2} \ =& \ \Gamma_x(E)^2 \sum_{y\in \Z^d} \left (1 + \frac{g^2 \rho(y)}{(E-\Omega)^2} \right ) |\hat G(y,x,E)|^2 \\ 
    =& \ \sum_{y\in \Z^d} \left (1 + \frac{g^2 \rho(y)}{(E-\Omega)^2} \right ) \frac{|\phi_E(y,x)|^2}{|\phi_E(x,x)|^2} \ = \ \frac{1}{\phi_E(x,x)} \ ,
    \end{split}
\end{equation}
where the second line holds at points $E\in \sigma(H)$ and follows from \eqref{eq:formofphiE}, \eqref{eq:spectrumcharacterization} and \eqref{eq:phiEnorm}.

Using \eqref{eq:spectmeas} and \eqref{eq:formofphiE}, we find the following formula for the modified correlator:
\begin{equation}\label{eq:QasG}
\begin{split}
    \tilde Q_\Lambda(y,x;W,s) \ &= \ \sum_{E\in \sigma(H)\cap I} W_g(E) |\phi_E(y,x)|^s |  |\phi_E(x,x)|^{1-s} \\
    &= \ \int_I |\hat G(y,x,E)|^s |\hat \Gamma_x(E)|^s d \nu_x(E) \ ,
    \end{split}
\end{equation}
where we have used \eqref{eq:spectrumcharacterization} to re-express the ratio $ (g^2\hat \rho_x- g^2\rho(x))/({E-\Omega})$ on the right hand side of \eqref{eq:formofphiE}, and the measure $d\nu_x(E)$ is given by
\begin{align}
    d\nu_x(E) \ = \  W_g(E)  \delta\left(\Gamma(E) - \frac{g^2\rho(x)-g^2\hat\rho_x}{E-\Omega}\right)dE\, .
\end{align}

We now average \eqref{eq:QasG} over the disorder, first computing the average with respecting to $\rho(x)$, noting that the modified Green's function $\hat G$ and $\Gamma_x(E)$ are independent of this variable.  We note that averaging over $\rho(x)$ is equivalent to conditioning on $\{\rho(y)\}_{y\neq x}$. Thus the resulting conditional expectation can be estimated as
\begin{align}
\begin{split}
   \E\left[ \vert \hat G(y,x,\right.&E\left.)\vert^{s} \vert\Gamma_x(E)\vert^{x}\delta\left(\Gamma(E) - \frac{g^2\rho(x)-g^2\hat\rho_x}{E-\Omega}\right) \middle | \{ \rho(y) \}_{y\neq 0} \right] \\ 
   =& \ \vert \hat G(y,x,E)\vert^{s}\frac{1}{2g^2\rho_0}  \int_0^{2g^2\rho_0} \left\vert\frac{\lambda-g^2\hat\rho_x}{E-\Omega}\right\vert^{s}   \delta\left(\Gamma_x(E) - \frac{\lambda -g^2\hat\rho_x}{E-\Omega}\right)d\lambda \\ 
   \le & \ \vert \hat G(y,x,E)\vert^{s} \left\vert\frac{2g^2\rho_0 }{E-\Omega}\right\vert^{s} 
   \frac{|E-\Omega|}{2g^2\rho_0} \int_0^{2g^2\rho_0} \delta\left(\Gamma_x(E)(E-\Omega) - \lambda+g^2\hat\rho_x \right) d\lambda \\
   \le & \ \left\vert\frac{E-\Omega}{2g^2\rho_0} \right\vert^{1-s} \vert \hat G(y,x,E)\vert^{s} \, .
   \end{split}
\end{align}
Averaging over the remaining random variables, as well as over $\hat\rho$ (with a distribution that is uniform in $[0,2\rho_0]$ and independent of all other variables), we find that 
\begin{equation}
    \E\left[ \tilde Q_\Lambda(y,x;W,s) \right]
    \ \leq  \ \int_I  \E\left [\vert G(y,x,E)\vert^{s} \right ] \left\vert\frac{E-\Omega}{2g^2\rho_0} \right\vert^{1-s} W_g(E)dE  \ .
\end{equation}

Returning to the correlator $Q_\Lambda(y,x;W)$, using \eqref{eq:QLambdaub}, \eqref{eq:modifiedcorrub}, the Cauchy-Schwartz inequality and the symmetry of the Green's function ($|G(y,x,E)|=|G(x,y,E)|$), we find that
\begin{equation}\label{eq:finalbound}
\begin{split}
    \E\left [ Q_\Lambda(y,x;W) \right ] \ & \le \ \sqrt{\E\left [\tilde Q_\Lambda(y,x;W,s) \right]\E\left [\tilde Q_\Lambda(x,y;W,s) \right]} \\ & \le \ \int_I  \E\left [\vert G_\Lambda(y,x,E)\vert^{s} \right ] \left\vert\frac{E-\Omega}{2g^2\rho_0} \right\vert^{1-s} W_g(E)dE  \ .
    \end{split}
\end{equation}
With this estimate we are now ready to prove Theorem \ref{maintheorem}.

\begin{proof}[Proof of \Cref{maintheorem}]
We see from \eqref{eq:finalbound} that
\begin{equation}
    \sup_{\Lambda\ssubset\Z^d} \E\left [ Q_\Lambda(y,x;W) \right ] \ \le \ \sup_{\Lambda \ssubset \Z^d} \int_I  \E\left [\vert G_\Lambda(y,x,E)\vert^{s} \right ] \left\vert\frac{E-\Omega}{2g^2\rho_0} \right\vert^{1-s} W_g(E)dE \ .
\end{equation}
Putting together the inequality $Q(y,x;W)\le 2$, Eq.\ \eqref{eq:finalbound} and \Cref{Greensfunctionbound}, we find that
\begin{align}
    \sum_{y\in\Z^d} \sup_{\Lambda \ssubset \Z^d} \E[Q_\Lambda(y,x;W)^2] \ 
    &\leq  \  C \int_I  \vert E-\Omega\vert^{1+\eps}\left ( 1 + \frac{g^2\rho_0}{|E-\Omega|^2} \right )  dE \ < \ \infty \ ,
    \end{align}
where we have used the fact that the interval $I$ is of the form
\begin{align}
    I = \{E\ | \ \vert E-\Omega\vert < \lambda(T)^{-1}g^2\rho_0 \}\ .
\end{align}
Since $Q(y,x;W)\le \liminf_{\Lambda \uparrow \Z^d} Q_\Lambda(y,x;W)$ (see \cite[Proposition 7.6]{aizenman2015random}), it follows that 
\begin{equation}
\sum_{y\in \Z^d} \E[Q(y,x;W)^2] \ < \ \infty \ .
\end{equation}
Thus by Lemma \ref{lem:dl->sl}, there is no continuous spectrum on the set $I\setminus \{\Omega\}$.  Since a continuous measure cannot be supported at the single point $\Omega$, the theorem follows.
\end{proof}

\section{Bosonic Multiple Photon systems}\label{sec:bosonic}
If the atomic field operators discussed in ~\cite{Kraisler_22,Kraisler_23} are considered to be bosonic instead of fermionic,
then we may easily consider the case of higher numbers of excitations. It follows that 
a system with $n$-excitations is equivalent to $n$ noninteracting single-excitation systems and that the $n$-excitation Hamiltonian $H_n: \ell^2(\Z^{nd})^{2^n}\to \ell^2(\Z^{nd})^{2^n}$ can be expressed as a tensor sum of the form
\begin{align}\label{eq:nexcitation}
    H_n = \sum_{k=1}^n H_{1,k}\ ,
\end{align}
where the operators $H_{1,k}$ are defined by
\begin{align}
    H_{1,k} = I\otimes \cdots I\otimes H\otimes I\cdots\otimes I\ .
\end{align}
The spectrum of such an operator is given by the Minkowski sum
\begin{align}
    \sigma(H_n) = \overline{\sum_{k=1}^n\sigma(H)}\ .
\end{align}
For example, the two photon Hamiltonian $H_2: \ell^2(\Z^d)^4\to \ell^2(\Z^d)^4$ can be written in the form
\begin{align}
    H_2 = H\otimes I+I\otimes H=\begin{pmatrix}\label{eq:twoexcitationH}
        T_{\bx}+T_{\by} & g\sqrt{\rho(\bx)} & g\sqrt{\rho(\by)} & 0 \\
        g\sqrt{\rho(\bx)} & T_{\by}+\Omega & 0 & g\sqrt{\rho(\by)} \\
        g\sqrt{\rho(\by)} & 0 & T_{\bx}+\Omega & g\sqrt{\rho(\bx)} \\
        0 & g\sqrt{\rho(\by)} & g\sqrt{\rho(\bx)} & 2\Omega
    \end{pmatrix}\ .
\end{align}
Putting together this decomposition of $H_n$ and \Cref{cor:fullypp}, we obtain the following result.
\begin{theorem}
    Consider the two excitation Hamiltonian $H_n$ given in \eqref{eq:nexcitation}, where $\{\rho(\bx)\}_{\bx\in\Z^d}$ are i.i.d. random variables uniformly distributed in $[0,2\rho_0]$. If the operator $T$ satisfies the hypothesis \eqref{eq:summability}, then for $g^2\rho_0$ large enough to satisfy the conclusion of \Cref{cor:fullypp} and almost every realization of the random variables $\{\rho(x)\}$, the operator $H_n$ has only pure point spectrum.
\end{theorem}

\section{Discussion}
We have considered a discrete model for the propagation of a one-photon state in a random medium of two-level atoms. In this setting, we have shown that the Hamiltonian of the system has only pure point spectrum in an interval centered at the resonant frequency of the atoms. The interval may be increased in size by increasing the average density of the atoms or the field-matter coupling, until it encompasses the entire spectrum of the Hamiltonian. Additionally, if the atoms are bosonic, we find that our results can be extended to the multi-photon setting. In future work, we plan to study the case of fermionic atoms and spin systems with mixed bosonic-fermionic character.
In both cases, the multi-particle system is interacting and the simple structure observed in \Cref{eq:nexcitation} no longer holds. Substantial work has been done on localization in many body quantum systems~\cite{Warzel_09, Chulaevsky_09} as well as in spin chains~\cite{Beaud_17, Bolter_22, Elgart_18, Elgart_22},  from which we believe key ideas may be applied.

\section*{Declaration of competing interest}
The authors declare no competing interests.

\section*{Data availability}
No data were generated or analyzed in this paper.

\section*{Acknowledgements}
We are grateful to Zaher Hani for valuable discussions. J.C.S. was supported, in part, by NSF Grant No. DMS-1912821 and AFOSR Grant No. FA9550-19-1-0320. J.H.S. was supported, in part, by NSF Grant No. DMS-2153946. J.E.K. was supported, in part, by Simons Foundation Math + X Investigator Award No. 376319 (Michael I. Weinstein).

\bibliographystyle{unsrt}
\bibliography{paper}

\end{document}